\documentclass[11pt]{article}
\usepackage{amsmath,amssymb,amsbsy,amsfonts, amsopn,amstext,amsxtra,
            euscript,amscd,graphicx}
\usepackage[paper=a4paper, left=1.95cm, right=1.95cm,
                            top=1.95cm, bottom=1.95cm, includefoot, centering]{geometry}
\usepackage{float,bm,setspace}
\usepackage[colorlinks=true]{hyperref}

\newtheorem{example}{Example}[section]
\newtheorem{definition}{Definition}[section]
\newtheorem{lemma}{Lemma}[section]
\newtheorem{theorem}{Theorem}[section]
\newtheorem{corollary}{Corollary}[section]
\newenvironment{proof}{\vspace{8pt}
\noindent{\bf Proof}: }{{\hfill {\large $\Box$}} \vspace{8pt}}

\begin{document}

\title{On the Eigenvalues of Certain Matrices Over $\mathbb{Z}_m$}

\author{ Liang Feng Zhang\\Nanyang Technological University\\s070006@e.ntu.edu.sg
}

\date{}

\maketitle

\begin{abstract}
Let $m,n>1$ be integers and $\mathbb{P}_{n,m}$ be the point set of the projective $(n-1)$-space (defined by [2]) over the ring $\mathbb{Z}_m$of integers modulo $m$. Let $A_{n,m}=(a_{uv})$ be the matrix with rows and columns
being labeled by elements of $\mathbb{P}_{n,m}$, where $a_{uv}=1$ if the inner product $\langle u,v \rangle=0$ and $a_{uv}=0$ otherwise. Let
$B_{n,m}=A_{n,m}A_{n,m}^t$. The eigenvalues of $B_{n,m}$ have been studied by [1, 2, 3], where their applications
in the study of expanders and locally decodable codes were described. In this paper, we completely
determine the eigenvalues of $B_{n,m}$ for general integers $m$ and $n$.

\end{abstract}

\section{Introduction}

Let $m>1$ be  an integer. Let $\mathbb{Z}_m$ be the ring of integers modulo $m$ and
$\mathbb{Z}_m^*$ be the group of units of $\mathbb{Z}_m$.
Let $n>1$ be an integer and
$\mathbb{Z}_m^n$ be the set of $n$-tuples with entries in $\mathbb{Z}_m$.
We say that $u, v\in \mathbb{Z}_m^n$	 are equivalent
(and write $u\sim v$) if
 there is a $\lambda\in \mathbb{Z}_m^*$ such that
$u_i=\lambda v_i$
for every $i\in [n]$.
If $u$ is not equivalent to $v$, we write $u\not \sim v$.
Let
$
\mathbb{S}_{n,m}=\{u\in \mathbb{Z}_m^n: \gcd(u_1,u_2,\ldots, u_n,m)=1\}
$
and
\begin{equation}
\mathbb{P}_{n,m}\triangleq \mathbb{S}_{n,m}/\sim
\end{equation}
be the set of equivalence classes of elements of $\mathbb{S}_{n,m}$ under $\sim$.
Let $A_{n,m}=(a_{uv})$ be the matrix with rows and columns being labeled by
elements of $\mathbb{P}_{n,m}$, where
$a_{uv}=1$ if the inner product $\langle u,v \rangle=0$ and $a_{uv}=0$ otherwise.
Let  $B_{n,m}=A_{n,m}A_{n,m}^t$.
For every $u\in \mathbb{P}_{n,m}$, let  $N(u)=\{v\in \mathbb{P}_{n,m}: \langle u,v \rangle=0 \}$  be the
neighborhood of $u$.
Let  $\theta_{n,m}=|\mathbb{P}_{n,m}|$.
Chee et al.~\cite{CL93} showed that
\begin{equation}
\label{eqn:theta}
\theta_{n,m}=m^{n-1}\prod_{p|m}
\left(1+\frac{1}{p}+\cdots+\frac{1}{p^{n-1}}\right)
\end{equation}
and $|N(u)|=\theta_{n-1,m}$
for every $u\in \mathbb{P}_{n,m}$.

Let $a,b$ be  positive integers. We denote by
 $I_a$  the  identity  matrix of order $a$. We denote by  $O_{a \times b}$ and
 $J_{a\times b}$  the $a\times b$  all-zero and all-one matrices, respectively. In particular, we  write
$O_a$ and $J_a$ when $a=b$, and
 write $I,~O$ and $J$ when $a,b$ are obvious.
When $m$ is a prime, Alon \cite{Alo86} showed that $B_{n,m}$ has two distinct eigenvalues
$\theta_{n-1,m}^2$  and
$m^{n-2}$.
The eigenvectors  with eigenvalue $\theta_{n-1,m}^2$ is the single column of $J_{l\times 1}$ and the eigenvectors   with
eigenvalue $m^{n-2}$ are the columns of the matrix
\begin{equation}
\label{eqn:K_d}
R_d=
\begin{pmatrix}
I_{d}\\
-J_{1\times d}
\end{pmatrix},
\end{equation}
where $l=\theta_{n,m}$ and $d=\theta_{n,m}-1$.
When $m=pq$ for two distinct primes, the eigenvalues of
$B_{n,m}$ have been determined by Chee et al. \cite{Zhang12}
(see Lemma 3.2), which have improtant applications in the study of matching families
in $\mathbb{Z}_m^n$.  Our work in this paper is mainly motivated by
Chee et al. \cite{Zhang12}  for its potential applications  in the study of matching families.

\section{Results}

In this paper, we completely determine the eigenvalues of $B_{n,m}$ for any positive integer $m$.
First of all, we deal with the prime power case and show the following theorem.

\begin{theorem}\label{pro:eig_B_p_e_n_result}
{\em (Prime Power Case)}
Let $m=p^e$ for a prime $p$  and positive integers $e$ and $n$. Then the eigenvalues of
$B_{n,m}$ and their multiplicities are as follows:
\begin{table}[H]
\begin{center}
\em
\doublespacing
     \begin{tabular}{   | l | l |}
     \hline
     Eigenvalue & Multiplicity  \\ \hline
     $\lambda_1 = p^{2(e-1)(n-2)}\cdot \theta_{n-1,p}^2 $ & $d_1 = 1$  \\ \hline
$\lambda_2 =   p^{(2e-1)(n-2)}$ & $d_2 =  \theta_{n,p}-1$   \\
\hline
  $\lambda_s=  p^{(2e+1-s)(n-2)}$ & $d_s=(p^{n-1}-1)\theta_{n,p^{s-2}}$    \\ \hline
     \end{tabular}
\caption{The eigenvalues of the matrix $B_{n,m}$}
\label{tab:eig}
 \end{center}
\end{table}
\end{theorem}

Theorem  \ref{pro:eig_B_p_e_n_result} makes it very easy to determine the eigenvalues of $B_{n,m}$ for a general integer $m$.
We define the  tensor product of two matrices  $A=(a_{ij})$ and $B$ to be the
 block matrix
\begin{equation}
A\otimes B=(a_{ij}\cdot B).
\end{equation}
We say that $A\sim B$ if $A$ can be obtained from $B$ by applying the same permutation
to the rows and columns. Clearly,
if $A\sim B$, then $A$ and $B$ have the same eigenvalues.
The following lemma allows us to determine the eigenvalues of
$B_{n,m}$ via  Theorem \ref{pro:eig_B_p_e_n_result}.
\begin{lemma}\label{thm:B_m_n_decomposition}
{\em (Tensor Lemma)}
Let $m=m_1\cdots m_r=p_1^{e_1}\cdots p_r^{e_r}$ for distinct primes $p_1,\ldots, p_r$ and
positive integers
$e_1,\ldots, e_r$,
where $m_s=p_s^{e_s}$ for every $s\in [r]$.
Then we have that
\begin{equation}
B_{n,m}\sim B_{n,m_1} \otimes \cdots \otimes B_{n,m_r}.
\end{equation}
\end{lemma}

 \begin{theorem}\label{pro:eig_B_m_n}
{\em (General Case)}
Let $m=m_1\cdots m_r=p_1^{e_1}\cdots p_r^{e_r}$ for distinct primes $p_1,\ldots, p_r$ and
positive integers
$e_1,\ldots, e_r$,
where $m_s=p_s^{e_s}$ for every $s\in [r]$.
Let $\lambda_{s}$ be an eigenvalue of $B_{n,m_s}$ of
multiplicity $d_{s}$ for every
$s\in[r]$. Then $\lambda_1\cdots\lambda_r$
is an  eigenvalue of  $B_{n,m}$ of multiplicity
$d_{1}\cdots d_{r}$.
\end{theorem}

\section{The Prime Power Case}\label{sec:spectrum}

In this section, we determine the eigenvalues of $B_{n,m}$ whenever $m$ is a prime power.

\subsection{Linear Equation Systems Over $\mathbb{Z}_{p^e}$}

Let $a,b,c,d \in \mathbb{Z}_{p^e}$ for a prime $p$ and an integer $e>0$.
We consider the following equation system
\begin{equation}\label{equ:order_2_equation_system}
\begin{pmatrix}
a & b\\
c & d
\end{pmatrix}
\begin{pmatrix}
x \\
y
\end{pmatrix}
\equiv
\begin{pmatrix}
0\\
0
\end{pmatrix}
\bmod p^e.
\end{equation}
Let $\mathcal{N}$ be the number of pairs $(x,y)\in \mathbb{Z}_{p^e}^2$ that satisfy (\ref{equ:order_2_equation_system}).
Let
\begin{equation}
N=\gcd(ad-bc, p^e\cdot \gcd(a,b,c,d,p^e)).
\end{equation}

\begin{lemma}\label{lem:order_2_equation_system}
We have that $\mathcal{N}=N$.
\end{lemma}
\begin{proof}
Suppose that $\gcd(ad-bc,p^e)=p^f$  and $\gcd(a,b,c,d,p^e)=p^g$ for  $f,g\in \{0,1,\ldots,e\}$.

\begin{itemize}
\item[--]
If $g=e$, then $a=b=c=d\equiv 0\bmod p^e$. It follows that $\mathcal{N}=p^{2e}=N$.

\item[--]
If $g=0$, then  (w.l.o.g.) we may suppose that $\gcd(a,p)=1$.
Let $a^{-1}$ be the inverse of $a$ modulo $p^e$.
Then it is easy to see that $(x,y)\in
\mathbb{Z}_{p^e}^2$ satisfies (\ref{equ:order_2_equation_system}) if and only if
$$y\equiv 0 \bmod p^{e-f} {\rm~and~} x\equiv-a^{-1}by \bmod p^{e}.$$
Hence, $\mathcal{N}$ is equal to the number of the choices of
$y\in \mathbb{Z}_{p^e}$, which is equal to $p^f=N$.

\item[--]
 If $0<g<e$, then
there are $p^{-2g} N$ pairs $(x,y)\in
 \mathbb{Z}_{p^{e-g}}^2$ that satisfy
\begin{equation}\label{equ:reduced_order_2_equation_system}
p^{-g}\cdot \begin{pmatrix}
a & b \\
c & d
\end{pmatrix}
\begin{pmatrix}
x \\
y
\end{pmatrix}
\equiv
\begin{pmatrix}
0\\
0
\end{pmatrix}
\bmod p^{e-g}
\end{equation}
due to the previous case.
If  $(x,y)\in \mathbb{Z}_{p^e}^2$ is a solution for
 (\ref{equ:order_2_equation_system}), then
$$(x \bmod p^{e-g}, y \bmod p^{e-g})$$
 is a solution for (\ref{equ:reduced_order_2_equation_system}).
Conversely, if $(x,y)\in \mathbb{Z}_{p^{e-g}}^2$ is a solution for (\ref{equ:reduced_order_2_equation_system}), then
 $$(x+kp^{e-g}, y+lp^{e-g})$$ is  a solution of
  (\ref{equ:order_2_equation_system})
 for every  $k,l\in \mathbb{Z}_{p^g}$.
Hence,
$\mathcal{N}=p^{2g}\cdot (p^{-2g}\cdot N)=N$.
\end{itemize}
\end{proof}

\begin{definition}
{\em ($p$-adic Valuation)}
Let $p$ be a prime and  $t\in \mathbb{Z}$, we denote by
$\nu_p(t)$ the largest nonnegative integer $e$ such that $p^e|t$.
 In particular,  we set $\nu_p(0)=\infty$.
\end{definition}

\begin{lemma}
Let $p$ be a  prime and $s,t\in \mathbb{Z}$. Then
\begin{itemize}
\item[--] $\nu_p(t)\leq\infty$ and $\nu_p(t)=\infty$ if and only if $t=0$;
\item[--] $\nu_p(st)=\nu_p(s)+\nu_p(t)$;
\item[--] $\nu_p(s+t)\geq \min\{\nu_p(s), \nu_p(t)\}$ and the equality holds when $\nu_p(s)\neq \nu_p(t)$.
\end{itemize}
\end{lemma}
Let $u,v\in \mathbb{P}_{n,p^e}$ be arbitrary.
For every  $i,j\in [n]$, we define
\begin{equation}
\label{eqn:xij}
\xi_{ij}=u_iv_j-u_jv_i.
\end{equation}
Let
\begin{equation}
\label{eqn:defxi}
\alpha=\min\{\nu_p(\xi_{ij}): i,j\in [n]\}{\rm~and~}\xi=p^{\min(\alpha,e)}.
\end{equation}
Then  $\xi$ is the greatest common divisor of all the integers in $\{\xi_{ij}: i,j\in [n]\}$ and $p^e$.

\begin{lemma}\label{lem:good_coeff_1}
The following properties hold
\begin{enumerate}
\item[\em (a)]  There are distinct integers  $i,j\in [n]$ such that
\begin{equation}
\label{eqn:defij}
\gcd(u_i,u_j,v_i,v_j,p^e)=1{\rm~and~} \xi=\gcd(\xi_{ij},p^e).
\end{equation}
\item[\em (b)]  $\nu_p(\xi)<e$.
\end{enumerate}
\end{lemma}
\begin{proof}
Since $u\in \mathbb{P}_{n,p^e}$,  there is an integer $s\in [n]$ such that
$\gcd(u_s,p)=1$.

(a)
If $\alpha=\infty$, then we may take $i=s$ and then (\ref{eqn:defij}) follows.
From now on, we suppose that
$\alpha<\infty$.   If  (\ref{eqn:defij}) does not hold, then
we have that
$\gcd(u_i,u_j,v_i,v_j, p^e)>1$ for any $i,j\in [n]$ such that
  $\nu_p(\xi_{ij})=\alpha$.
Since $\gcd(u_s, p)=1$, this implies that $p^{\alpha+1}|\xi_{si}$ and $p^{\alpha+1}|\xi_{sj}$.
It follows that   $p^{\alpha+1}|(u_i\xi_{sj}- u_j\xi_{si})$, i.e., $p^{\alpha+1}|u_s\xi_{ij}$.
Since $\gcd(u_s, p)=1$, we have that $p^{\alpha+1}|\xi_{ij}$, which  is  a contradiction.

(b)
If  $\nu_p(\xi)=e$, then $\xi_{si} \equiv 0 \bmod p^e$ for every $i\in [n]$.
It follows that $v_i\equiv u_s^{-1}v_s u_i\bmod p^e$,
where $u_s^{-1}$ is the inverse of $u_s$ modulo $p^e$.
Note that $v_s\not \equiv 0\bmod p$ since otherwise we will have that $v\notin \mathbb{P}_{n,p^e}$.
It follows that $u\sim v$, which is a contradiction because $u$ and $v$ are distinct equivalence classes.
\end{proof}

We shall determine the $(u,v)$ entry of $B_{n,p^e}$, i.e., $b_{uv}$ in the remainder of
 this section.  Clearly, $b_{uv}$ is the number of $w\in \mathbb{P}_{n,p^e}$ that
 satisfy the following equation system:
\begin{equation}\label{equ:2_wise_sol}
\left\{
\begin{aligned}
\langle  u,  w \rangle \equiv 0\bmod p^e , \\
\langle  v,  w\rangle \equiv 0\bmod p^e.
\end{aligned}
\right.
\end{equation}
For every $g\in \{0,1,\ldots, e\}$, we define
$$p^g\cdot  \mathbb{Z}_{p^e}^{n}=\{w\in \mathbb{Z}_{p^e}^{n}:
p^g|\gcd(w_1, \ldots, w_n)\},$$

\begin{lemma}\label{lem:2_wise_sol_class}
For every $g\in \{0,1,\ldots, e\}$,  the equation system
(\ref{equ:2_wise_sol}) has exactly  $p^{\beta+(e-g)(n-2)}$ solutions in
$p^g\cdot  \mathbb{Z}_{p^e}^{n}$, where $\beta=\min(\nu_p(\xi),e-g)$.
\end{lemma}

\begin{proof}
We prove for  Case I: $g=0$ and Case II: $0<g\leq e$, respectively.

\begin{itemize}
\item[--]
{Case I:}
Due to Lemma \ref{lem:good_coeff_1}, we may suppose that   $\gcd(u_1,u_2,v_1,v_2,p^e)=1$
and $\xi=\gcd(\xi_{12},p^e)$.
It suffices to show  that for every $(w_3,\ldots, w_n)\in \mathbb{Z}_{p^e}^{n-2}$,
the following equation system
\begin{equation}\label{equ:inter_order_2_es}
\begin{pmatrix}
u_1 & u_2\\
v_1 & v_2
\end{pmatrix}
\begin{pmatrix}
x\\
y
\end{pmatrix}
\equiv
\begin{pmatrix}
-\sum_{k=3}^{n} u_kw_k\\
-\sum_{k=3}^{n} v_k w_k
\end{pmatrix}
\bmod p^e
\end{equation}
has exactly $\xi$ solutions in $\mathbb{Z}_{p^e}^2$.
Due to Lemma \ref{lem:order_2_equation_system}, the homogenous form of
(\ref{equ:inter_order_2_es}) has  exactly $\xi$
solutions in $\mathbb{Z}_{p^e}^2$.

Hence, it suffices to show that
(\ref{equ:inter_order_2_es}) is solvable.	
Let $\nu_p(\xi_{12})=h$. Then there is an  integer $\eta \in \mathbb{Z}_{p^e}^*$
such that  $\xi_{12}=p^h \eta$.
Let $\eta^{-1}$ be the inverse of $\eta$ modulo $p^e$.
 Then
$$(x,y)=\bigg(\sum_{k=3}^{n} \frac{\eta^{-1} \xi_{2k}}{p^h}w_k,
\sum_{k=3}^{n} \frac{\eta^{-1} \xi_{k1}}{p^h}w_k\bigg)$$
is a solution of (\ref{equ:inter_order_2_es}).

\item[--] Case II:
Let $w\in p^g\cdot \mathbb{Z}_{p^e}^{n}$ and $w^\prime=p^{-g}\cdot w$.
Then $w$  satisfies  (\ref{equ:2_wise_sol}) if and only if
\begin{equation}\label{equ:2_wise_sol_class}
\left\{
\begin{aligned}
\langle  u,  w^\prime \rangle \equiv 0\bmod p^{e-g} , \\
\langle  v,  w^\prime \rangle \equiv 0\bmod p^{e-g}.
\end{aligned}
\right.
\end{equation}
Due to Case I,  the equation system (\ref{equ:2_wise_sol_class})
has exactly $p^{\beta+(e-g)(n-2)}$ solutions  $w^\prime\in \mathbb{Z}_{p^{e-g}}^{n}$.
\end{itemize}
\end{proof}

\begin{lemma}\label{cor:B_entry}
Let $u,v\in \mathbb{P}_{n,m}$ be arbitrary. Then the $(u,v)$ entry of $B_{n,p^e}$ is
\begin{equation}
\label{eqn:buv}
b_{uv}=\frac{1}{\phi(p^e)}\left(p^{\nu_p(\xi)+e(n-2)}-p^{\min(\nu_p(\xi),e-1)+
(e-1)(n-2)}\right).
\end{equation}
\end{lemma}

\begin{proof}
Clearly, we have that  $\mathbb{Z}_{p^e}^n\setminus p\cdot \mathbb{Z}_{p^e}^n=\mathbb{S}_{n,p^e}$.
Let $N_{uv}$ be the number of solutions of (\ref{equ:2_wise_sol}) in $\mathbb{S}_{n,p^e}$. Due to Lemma
\ref{lem:2_wise_sol_class}, it is not hard to see that
\begin{equation*}
N_{uv}=
\left\{
\begin{aligned}
p^{\nu_p(\xi)+e(n-2)}-p^{\min(\nu_p(\xi),e-1)+(e-1)(n-2)} \hspace{0.5cm} &{\rm~if~} u=v, \\
p^{e(n-1)}-p^{(e-1)(n-1)} \hspace{3.65cm} &{\rm~if~} u\neq v.
\end{aligned}
\right.
\end{equation*}
Note that $\nu_p(\xi)=e$ when $u=v$ and each equivalence class in
$\mathbb{P}_{n,p^e}$ contains exactly $\phi(p^e)$ elements of $\mathbb{S}_{n,p^e}$.
It follows that $b_{uv}=\phi(p^e)^{-1}N_{uv}$ and
the equation
(\ref{eqn:buv}) follows.
\end{proof}

\subsection{The Eigenvalues}

We proceed to determine the
eigenvalues of $B_{n,p^e}$.
Let $\sigma: \mathbb{S}_{n,p^e}\rightarrow
\mathbb{S}_{n, p^{e-1}}$
be the mapping defined by
\begin{equation}
\sigma(u)=(u_1\bmod p^{e-1}, \ldots, u_n\bmod p^{e-1}),
\end{equation}
where $u\in \mathbb{S}_{n,p^e}$.
Let $\tau: \mathbb{S}_{n,p^{e-1}}\rightarrow \mathbb{P}_{n,p^{e-1}}$ to be the
mapping such that
\begin{equation}
\tau(u)={\rm the~equivalence~class~of~}u.
\end{equation}
\begin{center}
\begin{figure}[H]
\centering
\ifpdf
  \setlength{\unitlength}{1bp}%
  \begin{picture}(196.70, 112.41)(0,0)
  \put(0,0){\includegraphics{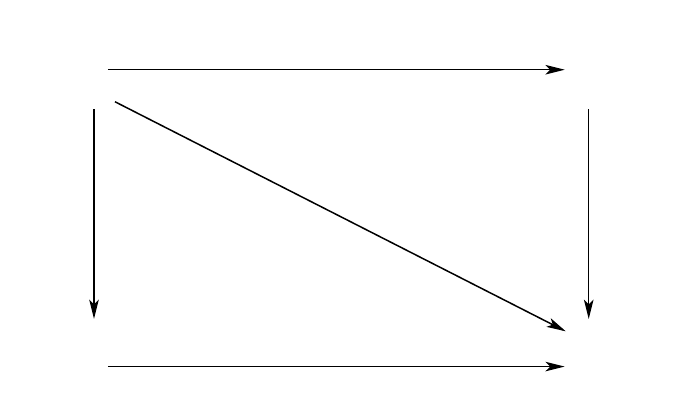}}
  \put(5.67,89.57){\fontsize{9.72}{11.66}\selectfont $\mathbb{S}_{n,p^e}$}
  \put(165.89,89.57){\fontsize{9.72}{11.66}\selectfont $\mathbb{S}_{n,p^{e-1}}$}
  \put(167.74,7.76){\fontsize{9.72}{11.66}\selectfont $\mathbb{P}_{n,p^{e-1}}$}
  \put(87.68,99.15){\fontsize{9.72}{11.66}\selectfont $\sigma$}
  \put(87.68,60.68){\fontsize{9.72}{11.66}\selectfont $\rho$}
  \put(173.62,60.68){\fontsize{9.72}{11.66}\selectfont $\tau$}
  \put(191.03,60.68){\fontsize{9.72}{11.66}\selectfont }
  \put(5.67,7.76){\fontsize{9.72}{11.66}\selectfont $\mathbb{P}_{n,p^e}$}
  \put(87.68,13.66){\fontsize{9.72}{11.66}\selectfont $\delta$}
  \put(15.30,60.68){\fontsize{9.72}{11.66}\selectfont $\psi$}
  \end{picture}%
\else
  \setlength{\unitlength}{1bp}%
  \begin{picture}(196.70, 112.41)(0,0)
  \put(0,0){\includegraphics{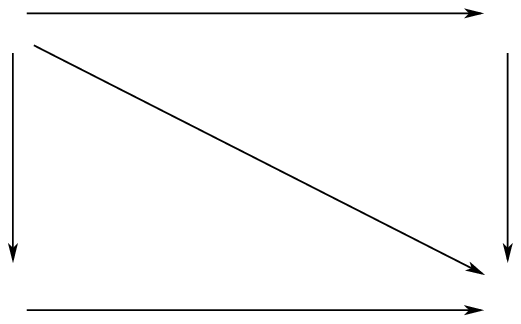}}
  \put(5.67,89.57){\fontsize{9.72}{11.66}\selectfont $\mathbb{S}_{n,p^e}$}
  \put(165.89,89.57){\fontsize{9.72}{11.66}\selectfont $\mathbb{S}_{n,p^{e-1}}$}
  \put(167.74,7.76){\fontsize{9.72}{11.66}\selectfont $\mathbb{P}_{n,p^{e-1}}$}
  \put(87.68,99.15){\fontsize{9.72}{11.66}\selectfont $\sigma$}
  \put(87.68,60.68){\fontsize{9.72}{11.66}\selectfont $\rho$}
  \put(173.62,60.68){\fontsize{9.72}{11.66}\selectfont $\tau$}
  \put(191.03,60.68){\fontsize{9.72}{11.66}\selectfont }
  \put(5.67,7.76){\fontsize{9.72}{11.66}\selectfont $\mathbb{P}_{n,p^e}$}
  \put(87.68,13.66){\fontsize{9.72}{11.66}\selectfont $\delta$}
  \put(15.30,60.68){\fontsize{9.72}{11.66}\selectfont $\psi$}
  \end{picture}%
\fi
\caption{The mappings}
\end{figure}
\end{center}
Let $\rho=\tau \circ \sigma$.
The following lemma shows that the mapping $\rho$ is balanced in the sense that
every equivalence class in $\mathbb{P}_{n,p^{e-1}}$ has the same number of
preimages in $\mathbb{S}_{n,p^e}$.

\begin{lemma}
\label{lem:rho}
We have that
 $|\rho^{-1}(v)|=p^{n}\phi(p^{e-1})$ for every $v\in \mathbb{P}_{n,p^{e-1}}$.
\end{lemma}
\begin{proof}
Let $w\in \mathbb{S}_{n,p^{e-1}}$ be arbitrary.
For every $\gamma=(\gamma_1,\ldots, \gamma_n)\in \mathbb{Z}_p^n$, we define
$$w_\gamma=(w_1+\gamma_1 p^{e-1},\ldots, w_n+\gamma_n p^{e-1}).$$
It is easy to see that
$w_\gamma\in \rho^{-1}(w)$ for every $\gamma\in \mathbb{Z}_p^n$ and
$w_\gamma\neq w_{\gamma^\prime}$ whenever $\gamma,\gamma^\prime\in \mathbb{Z}_p^n$
are distinct.
It follows that $|\sigma^{-1}(w)|\geq p^n$.
If $|\sigma^{-1}(w)|>p^n$ for certain choice of $w\in \mathbb{S}_{n,p^{e-1}}$, then we must have that
$$|\mathbb{S}_{n,p^e}| >  p^n |\mathbb{S}_{n,p^{e-1}}|.$$
However, due to (\ref{eqn:theta}), we have that
\begin{equation*}
|\mathbb{S}_{n,p^e}|=\phi(p^e)|\mathbb{P}_{n,p^{e}}|=
p^n \phi(p^{e-1})|\mathbb{P}_{n,p^{e-1}}|=p^n|\mathbb{S}_{n,p^{e-1}}|.
\end{equation*}
Therefore, we must have that $|\sigma^{-1}(w)|=p^n$.
Let $v\in \mathbb{P}_{n,p^{e-1}}$ be arbitrary. Then it is easy to see that
$|\tau^{-1}(v)|=\phi(p^{e-1})$. It follows that
$|\rho^{-1}(v)|=|\sigma^{-1}(\tau^{-1}(v))|=p^n\phi(p^{e-1})$.
\end{proof}

Let $\psi: \mathbb{S}_{n,p^e}\rightarrow
\mathbb{P}_{n, p^{e}}$
be the mapping defined by
\begin{equation}
\psi(u)={\rm~the~equivalence~class~of~}u,
\end{equation}
where $u\in \mathbb{S}_{n,p^e}$.
Let $\delta: \mathbb{P}_{n,p^e}\rightarrow
\mathbb{P}_{n, p^{e-1}}$
be the mapping defined by
\begin{equation}
\delta(u)={\rm~the~equivalence~class~of~}(u_1\bmod p^{e-1}, \ldots, u_n\bmod p^{e-1}),
\end{equation}
where $u\in \mathbb{P}_{n,p^e}$. Then
$\delta\circ \psi=\rho$. The following lemma shows that
the mapping $\delta$ is $p^{n-1}$ to 1.
\begin{lemma}
\label{lem:partition}
We have that $|\delta^{-1}(v)|=p^{n-1}$ for every $v\in \mathbb{P}_{n,p^{e-1}}$.
\end{lemma}
\begin{proof}
Suppose that $u,v\in \mathbb{S}_{n,p^e}$ and $u\sim v$. Then there is an integer $\lambda\in
\mathbb{Z}_{p^e}^*$ such that
$u_i\equiv \lambda v_i\bmod p^e$ for every $i\in [n]$.
It follows that $u_i\equiv \lambda v_i\bmod p^{e-1}$ for the integer
$\lambda\in \mathbb{Z}_{p^{e-1}}^*$. In other words, we have that
$\rho(u)=\rho(v)$. Hence, for every $v\in \mathbb{P}_{n,p^{e-1}}$,
the set $\rho^{-1}(v)$ is the union of  disjoint equivalence classes
of the elements of $\mathbb{S}_{n,p^e}$.
Due to Lemma \ref{lem:rho}, we have that $|\rho^{-1}(v)|=p^n\phi(p^{e-1})$.
Since the elements in $\rho^{-1}(v)$ consists of
$\phi(p^e)^{-1} p^n \phi(p^{e-1})=p^{n-1}$ equivalence classes,
we have that
$$|\delta^{-1}(v)|={\rm the~number~of~equivalence~classes~contained~by~}\rho^{-1}(v)=p^{n-1},$$
which is the expected result.
\end{proof}


\begin{figure}[H]
\begin{center}
$
\begin{array}
{c||ccccccc|ccccccc|ccccccc|ccccccc}
(001)& 6& 1& 1& 1& 1& 1& 1& 2& 1& 1& 1& 1& 1& 1& 2& 1& 1& 1& 1& 1& 1& 2& 1& 1& 1& 1& 1& 1 \\
(010)& 1& 6& 1& 1& 1& 1& 1& 1& 2& 1& 1& 1& 1& 1& 1& 2& 1& 1& 1& 1& 1& 1& 2& 1& 1& 1& 1& 1 \\
(011)& 1& 1& 6& 1& 1& 1& 1& 1& 1& 2& 1& 1& 1& 1& 1& 1& 2& 1& 1& 1& 1& 1& 1& 2& 1& 1& 1& 1 \\
(100)& 1& 1& 1& 6& 1& 1& 1& 1& 1& 1& 2& 1& 1& 1& 1& 1& 1& 2& 1& 1& 1& 1& 1& 1& 2& 1& 1& 1 \\
(101)& 1& 1& 1& 1& 6& 1& 1& 1& 1& 1& 1& 2& 1& 1& 1& 1& 1& 1& 2& 1& 1& 1& 1& 1& 1& 2& 1& 1 \\
(110)& 1& 1& 1& 1& 1& 6& 1& 1& 1& 1& 1& 1& 2& 1& 1& 1& 1& 1& 1& 2& 1& 1& 1& 1& 1& 1& 2& 1 \\
(111)& 1& 1& 1& 1& 1& 1& 6& 1& 1& 1& 1& 1& 1& 2& 1& 1& 1& 1& 1& 1& 2& 1& 1& 1& 1& 1& 1& 2 \\ \hline
(021)& 2& 1& 1& 1& 1& 1& 1& 6& 1& 1& 1& 1& 1& 1& 2& 1& 1& 1& 1& 1& 1& 2& 1& 1& 1& 1& 1& 1 \\
(012)& 1& 2& 1& 1& 1& 1& 1& 1& 6& 1& 1& 1& 1& 1& 1& 2& 1& 1& 1& 1& 1& 1& 2& 1& 1& 1& 1& 1 \\
(013)& 1& 1& 2& 1& 1& 1& 1& 1& 1& 6& 1& 1& 1& 1& 1& 1& 2& 1& 1& 1& 1& 1& 1& 2& 1& 1& 1& 1 \\
(102)& 1& 1& 1& 2& 1& 1& 1& 1& 1& 1& 6& 1& 1& 1& 1& 1& 1& 2& 1& 1& 1& 1& 1& 1& 2& 1& 1& 1 \\
(103)& 1& 1& 1& 1& 2& 1& 1& 1& 1& 1& 1& 6& 1& 1& 1& 1& 1& 1& 2& 1& 1& 1& 1& 1& 1& 2& 1& 1 \\
(112)& 1& 1& 1& 1& 1& 2& 1& 1& 1& 1& 1& 1& 6& 1& 1& 1& 1& 1& 1& 2& 1& 1& 1& 1& 1& 1& 2& 1 \\
(113)& 1& 1& 1& 1& 1& 1& 2& 1& 1& 1& 1& 1& 1& 6& 1& 1& 1& 1& 1& 1& 2& 1& 1& 1& 1& 1& 1& 2 \\ \hline
(201)& 2& 1& 1& 1& 1& 1& 1& 2& 1& 1& 1& 1& 1& 1& 6& 1& 1& 1& 1& 1& 1& 2& 1& 1& 1& 1& 1& 1 \\
(210)& 1& 2& 1& 1& 1& 1& 1& 1& 2& 1& 1& 1& 1& 1& 1& 6& 1& 1& 1& 1& 1& 1& 2& 1& 1& 1& 1& 1 \\
(211)& 1& 1& 2& 1& 1& 1& 1& 1& 1& 2& 1& 1& 1& 1& 1& 1& 6& 1& 1& 1& 1& 1& 1& 2& 1& 1& 1& 1 \\
(120)& 1& 1& 1& 2& 1& 1& 1& 1& 1& 1& 2& 1& 1& 1& 1& 1& 1& 6& 1& 1& 1& 1& 1& 1& 2& 1& 1& 1 \\
(121)& 1& 1& 1& 1& 2& 1& 1& 1& 1& 1& 1& 2& 1& 1& 1& 1& 1& 1& 6& 1& 1& 1& 1& 1& 1& 2& 1& 1 \\
(130)& 1& 1& 1& 1& 1& 2& 1& 1& 1& 1& 1& 1& 2& 1& 1& 1& 1& 1& 1& 6& 1& 1& 1& 1& 1& 1& 2& 1 \\
(131)& 1& 1& 1& 1& 1& 1& 2& 1& 1& 1& 1& 1& 1& 2& 1& 1& 1& 1& 1& 1& 6& 1& 1& 1& 1& 1& 1& 2 \\ \hline
(221)& 2& 1& 1& 1& 1& 1& 1& 2& 1& 1& 1& 1& 1& 1& 2& 1& 1& 1& 1& 1& 1& 6& 1& 1& 1& 1& 1& 1 \\
(212)& 1& 2& 1& 1& 1& 1& 1& 1& 2& 1& 1& 1& 1& 1& 1& 2& 1& 1& 1& 1& 1& 1& 6& 1& 1& 1& 1& 1 \\
(213)& 1& 1& 2& 1& 1& 1& 1& 1& 1& 2& 1& 1& 1& 1& 1& 1& 2& 1& 1& 1& 1& 1& 1& 6& 1& 1& 1& 1 \\
(122)& 1& 1& 1& 2& 1& 1& 1& 1& 1& 1& 2& 1& 1& 1& 1& 1& 1& 2& 1& 1& 1& 1& 1& 1& 6& 1& 1& 1 \\
(123)& 1& 1& 1& 1& 2& 1& 1& 1& 1& 1& 1& 2& 1& 1& 1& 1& 1& 1& 2& 1& 1& 1& 1& 1& 1& 6& 1& 1 \\
(132)& 1& 1& 1& 1& 1& 2& 1& 1& 1& 1& 1& 1& 2& 1& 1& 1& 1& 1& 1& 2& 1& 1& 1& 1& 1& 1& 6& 1 \\
(133)& 1& 1& 1& 1& 1& 1& 2& 1& 1& 1& 1& 1& 1& 2& 1& 1& 1& 1& 1& 1& 2& 1& 1& 1& 1& 1& 1& 6 \\
\end{array}
$
\end{center}
\caption{The matrix $B_{3,4}$}
\label{fig:B34}
\end{figure}
Due to Lemma \ref{lem:partition},  we can partition
$\mathbb{P}_{n,p^e}$ as $l=p^{n-1}$ disjoint subsets
$K_1,\ldots, K_l$ such that
\begin{equation}
|K_h\cap \delta^{-1}(u)|=1.
\end{equation}
 for every $h\in [l]$  and $u\in \mathbb{P}_{n,p^{e-1}}$.
 For every $a,b\in[l]$, we denote by
$C_{ab}$ a matrix with rows and columns being labeled by elements of $K_a$  and  $K_b$, respectively. For every $u\in K_a$ and
  $v\in K_b$, the $(u,v)$ entry of $C_{ab}$  is  defined to be
the number of solutions of (\ref{equ:2_wise_sol}) in $\mathbb{P}_{n,p^e}$.
Let $C=(C_{ab})$ be a block matrix. Then
\begin{equation}
C\sim B_{n,p^e}
\end{equation}

\begin{example}
\label{example}
We expalin the above description by an example. Let $n=3,~p=2$ and $e=2$. Then
simple calculations show that $\mathbb{P}_{3,4}$  consists of the following equivalence classes
\begin{equation}
\begin{tabular}{ cccccccc }
$K_1:~~$ 001 & 010 & 011 & 100 & 101 & 110 & 111 \\
$K_2:~~$ 021 & 012 & 013 & 102 & 103 & 112 & 113 \\
$K_3:~~$ 201 & 210 & 211 & 120 & 121 & 130 & 131 \\
$K_4:~~$ 221 & 212 & 213 & 122 & 123 & 132 & 133 \\
 \end{tabular}
\end{equation}
where $l=4$ and the $i$th row is the  $i$th subset of $\mathbb{P}_{3,4}$ for every $i\in \{1,2,3,4\}$. Clearly, we have that
 $\mathbb{P}_{3,2}=
\{ \textit{001,~010,~011,~100,~101,~110,~111} \}$. It is trivial to verify that
$|\delta^{-1}(u)\cap K_h|=1$ for every $u\in \mathbb{P}_{3,2}$ and $h\in \{1,2,3,4\}$.
Figure \ref{fig:B34} depicts the matrix $B_{3,4}$, where the rows are columns are labeled by
 elements of
$\mathbb{P}_{3,4}$. Actually, the $B_{3,4}$ is a $\textit{4} \times \textit{4}$ block matrix, where each block is a square matrix of
order 7. More precisely, we have that
\begin{equation}
C_{aa}=
\begin{pmatrix}
6& 1& 1& 1& 1& 1& 1\\
1& 6& 1& 1& 1& 1& 1\\
1& 1& 6& 1& 1& 1& 1\\
1& 1& 1& 6& 1& 1& 1\\
1& 1& 1& 1& 6& 1& 1\\
1& 1& 1& 1& 1& 6& 1\\
1& 1& 1& 1& 1& 1& 6\\
\end{pmatrix}
\end{equation}
for every $a\in \{1,2,3,4\}$ and
\begin{equation}
C_{ab}=
\begin{pmatrix}
2& 1& 1& 1& 1& 1& 1 \\
1& 2& 1& 1& 1& 1& 1 \\
1& 1& 2& 1& 1& 1& 1 \\
1& 1& 1& 2& 1& 1& 1 \\
1& 1& 1& 1& 2& 1& 1 \\
1& 1& 1& 1& 1& 2& 1 \\
1& 1& 1& 1& 1& 1& 2 \\
\end{pmatrix}
\end{equation}
for every $a,b\in \{1,2,3,4\}$ such that $a\neq b$.
It is trivial to verify that
 the $(u,v)$ entry of $C_{aa}$ and $C_{ab}$  is  equal to
the number of solutions of (\ref{equ:2_wise_sol}) in $\mathbb{P}_{3,4}$
for every $u\in K_a$ and $ v\in K_b$.
\end{example}

\begin{figure}[H]
\begin{center}
$
\begin{array}
{c||ccccccc}
(001)& 3& 1& 1& 1& 1& 1& 1\\
(010)& 1& 3& 1& 1& 1& 1& 1\\
(011)& 1& 1& 3& 1& 1& 1& 1\\
(100)& 1& 1& 1& 3& 1& 1& 1\\
(101)& 1& 1& 1& 1& 3& 1& 1\\
(110)& 1& 1& 1& 1& 1& 3& 1\\
(111)& 1& 1& 1& 1& 1& 1& 3
\end{array}
$
\end{center}
\caption{The matrix $B_{3,2}$}
\label{fig:B32}
\end{figure}
Note that the matrix $B_{3,2}$ can be depicted by   Figure \ref{fig:B32}.
The following lemma shows that there are connections between $B_{n,p^e}$
and $B_{n,p^{e-1}}$. The connections are clear when we partition the set
$\mathbb{P}_{n,p^e}$  as the disjoint subsets $K_1,\ldots, K_l$, where $\delta(K_h)=\mathbb{P}_{n,p^{e-1}}$
for every $h\in [l]$.
%
%

\begin{lemma}
\label{lem:diag2}
Let $a,b\in[l]$ and $a\neq b$.
Let $u\in K_a$ and $v\in K_b$ be such that $\delta(u)=\delta(v)=w\in \mathbb{P}_{n,p^{e-1}}$. Then the $(u,v)$ entry of $C_{ab}$ is equal to
\begin{equation}
\label{eqn:diag1}
\frac{1}{\phi(p^e)}\left(p^{e(n-1)-1}-p^{(e-1)(n-1)}\right).
\end{equation}
\end{lemma}
\begin{proof}
Without loss of generality, we may suppose that $\gcd(w_1,p)=1$.
Since $\delta(u)=\delta(v)=w$, there are integers
$x_2,\ldots, x_n, y_2,\ldots, y_n\in \mathbb{Z}_p$ such that
\begin{equation}
\label{eqn:wuv}
\begin{split}
u&\sim (w_1,w_2+x_2 p^{e-1}, \ldots, w_n+x_n p^{e-1}),\\
v&\sim(w_1,w_2+y_2 p^{e-1}, \ldots, w_n+y_n p^{e-1}).
\end{split}
\end{equation}
Let $\xi$ be defined by (\ref{eqn:defxi}).
Due to (\ref{eqn:wuv}), it is trivial to verify that
$\xi=p^{e-1}$.
It follows that
(\ref{eqn:diag1}) is an immediate consequence of (\ref{eqn:buv}).
\end{proof}

\begin{lemma}
\label{lem:notdiag}
Let $a,b\in [l]$ be arbitrary. Let $u\in K_a$ and $v\in K_b$ be such that $\delta(u)=u^\prime\neq v^\prime= \delta(v)$.
Then the $(u,v)$ entry of $C_{ab}$ is equal to the product of
$p^{n-3}$ and the $(u^\prime, v^\prime)$ entry of
$B_{n,p^{e-1}}$.
\end{lemma}
\begin{proof}
Clearly, there are integers $x_1, \ldots, x_{n}, 	y_1,\ldots,
y_n\in \mathbb{Z}_p$  such that
\begin{equation}
\begin{split}
u&\sim(u^\prime_1+x_1 p^{e-1},u^\prime_2+x_2 p^{e-1}, u^\prime_3+x_3 p^{e-1}, \ldots, u^\prime_n+x_n p^{e-1}),\\
v&\sim (v^\prime_1+y_1 p^{e-1}, v^\prime_2+y_2p^{e-1}, v^\prime_3+y_3 p^{e-1}, \ldots,
v^\prime_n+y_n p^{e-1}).
\end{split}
\end{equation}
 Let
$
\eta_{ij}=
u^\prime_iv^\prime_j-u^\prime_jv^\prime_i
$
for every  $i,j\in [n]$. Let $\xi_{ij}$ be defined by (\ref{eqn:xij}). Then
$\xi_{ij}\equiv \eta_{ij}\bmod p^{e-1}$. Since $u\not\sim v$,  Lemma \ref{lem:good_coeff_1} implies that
$ \nu_p(\xi)<e
$ and $ \nu_p(\eta)<e-1$.
We claim that  $ \nu_p(\xi)=\nu_p(\eta).$
In fact, by Lemma \ref{lem:good_coeff_1}, there are integers
$\tilde{i},\tilde{j}\in [n]$ and $\hat{i}, \hat{j}\in [n]$ such that
$$\xi=\gcd(\xi_{\tilde{i}\tilde{j}}, p^e){\rm~and~}\eta=\gcd(\eta_{\hat{i}\hat{j}}, p^{e-1}).$$
Since $\xi_{\hat{i}\hat{j}}\equiv \eta_{\hat{i}\hat{j}}\bmod p^{e-1}$, we
must have that $\nu_p(\xi_{\hat{i}\hat{j}}) \leq \nu_p(\eta_{\hat{i}\hat{j}})$.
It follows that
$$\nu_p(\xi)=
\nu_p(\xi_{\tilde{i}\tilde{j}})\leq \nu_p(\xi_{\hat{i}\hat{j}}) \leq \nu_p(\eta_{\hat{i}\hat{j}})=\nu_p(\eta).$$
Similarly, we must have that $\nu_p(\eta)\leq \nu_p(\xi)$. Hence,
our claim holds. In particular, we have that $\nu_p(\xi)=\nu_p(\eta)<e-1$.
Now the lemma is an immediate consequence of (\ref{eqn:buv}).
\end{proof}

\begin{lemma}
\label{lem:cab}
Let $p$ be a prime and $n,e>1$ be integers. Then for any intgers $a,b\in [l]$, we have that
\begin{equation}
\label{eqn:cab}
C_{ab}=
\left\{
\begin{split}
p^{n-3}B_{n,p^{e-1}}-p^{(e-1)(n-2)-1} I\hspace{2.6cm} &{\rm~if~} a\neq b, \\
p^{n-3}B_{n,p^{e-1}}+\left(p^{e(n-2)}-p^{(e-1)(n-2)-1}\right)I \hspace{0.5cm} &{\rm~if~} a=b.
\end{split}
\right.
\end{equation}
\end{lemma}
\begin{proof}
As depicted by Example \ref{example}, we partition the set $\mathbb{P}_{n,p^e}$ as
$l$ disjoint  subsets, where $l=p^{n-1}$.
Equation (\ref{eqn:buv}) shows that  the diagonal entries of
$B_{n,p^{e-1}}$ are all equal to
\begin{equation*}
\frac{1}{\phi(p^{e-1})}\left({p^{(e-1)(n-1)}-p^{(e-2)(n-1)}}\right).
\end{equation*}
When $a\neq b$,
 Lemma \ref{lem:diag2} shows that the diagonal entries of $C_{ab}$ are all equal to
 $$\frac{1}{\phi(p^e)}\left(p^{e(n-1)-1}-p^{(e-1)(n-1)}\right).
 $$
 Clearly, the difference between the diagonal entries is $p^{(e-1)(n-2)-1}$.
 Lemma \ref{lem:notdiag} shows that the non-diagonal entries of
 $C_{ab}$ are  $p^{n-3}$ times of those  of $B_{n,p^{e-1}}$.
It follows that
$$ C_{ab}=p^{n-3}B_{n,p^{e-1}}-p^{(e-1)(n-2)-1} I,$$
which is the first equality. Note that the diagonal entries of $C_{aa}$ are equal to
$$
\frac{1}{\phi(p^e)}\left(p^{e(n-1)}-p^{(e-1)(n-1)}\right).
$$
Similarly, we can prove the second  part of (\ref{eqn:cab}).
\end{proof}

\begin{lemma}\label{lem:eigenvalue_recursive_formula}
If $\lambda$ is an eigenvalue of $B_{n,p^{e-1}}$,
then $p^{2n-4}\cdot \lambda $ is an eigenvalue of $B_{n,p^e}$.
\end{lemma}

\begin{proof}
Due to Lemma \ref{lem:cab},
 $B_{n,p^e}
-p^{e(n-2)} I$  is an $l\times l$ block matrix, where each block
 is equal to
 $$C_{12}=p^{n-3}B_{n,p^{e-1}}-p^{(e-1)(n-2)-1} I.$$
 It follows that  $\mu=p^{n-3}\lambda-p^{(e-1)(n-2)-1}$ is an eigenvalue of $B_{n,p^e}$.
Thus
$$l \mu+p^{e(n-2)}=p^{n-1} \mu+p^{e(n-2)}=p^{2n-4}\lambda$$ is an eigenvalue of $B_{n,p^e}$.
\end{proof}

\begin{lemma}\label{lem:small_eig}
$\lambda_{e+1}=
p^{e(n-2)}$ is an eigenvalue of $B_{n,p^e}$ of
multiplicity
at least
$(p^{n-1}-1)\theta_{n,p^{e-1}}.$
\end{lemma}

\begin{proof}
Let $a\in [l-1]$. For every $u\in K_a$ and $v\in K_l$ such that
$\delta(u)=\delta(v)$, let
$w$
be a vector with coordinates being labeled by elements of  $\mathbb{P}_{n,p^e}$
such that the coordinate labeled by
$u$ is  1, the coordinate  labeled by $v$ is $-1$ and all the other coordinates are 0. Due to Lemma
\ref{lem:cab}, we have that
$B_{n,p^e}\cdot w=p^{e(n-2)}\cdot w$. Since there are
$$(l-1)\cdot |K_a|
=(p^{n-1}-1)\cdot \theta_{n,p^{e-1}}$$
choices for $w$ when $a$ is taken over $[l-1]$ and $u$ is taken over $K_a$ for every $a$.
Clearly, all the $w$'s are linearly
independent. The eigenvalue is of multiplicity at least $(p^{n-1}-1)\theta_{n,p^{e-1}}$.
\end{proof}

\begin{theorem}\label{pro:eig_B_p_e_n}
{\em (Prime Power Case)}
Let $m=p^e$ for a prime $p$  and positive integers $e$ and $n$. Then the eigenvalues of
$B_{n,m}$ and their multiplicities are as follows ($s\in \{3,\ldots,e+1\}$):
\begin{table}[H]
\begin{center}
\em
\doublespacing
     \begin{tabular}{   | l | l |}
     \hline
     Eigenvalue & Multiplicity  \\ \hline
     $\lambda_1 = p^{2(e-1)(n-2)}\cdot \theta_{n-1,p}^2 $ & $d_1 = 1$  \\ \hline
$\lambda_2 =   p^{(2e-1)(n-2)}$ & $d_2 =  \theta_{n,p}-1$   \\
\hline
  $\lambda_s=  p^{(2e+1-s)(n-2)}$ & $d_s=(p^{n-1}-1)\theta_{n,p^{s-2}}$    \\ \hline
     \end{tabular}
\caption{The eigenvalues of the matrix $B_{n,m}$}
\label{tab:eig}
 \end{center}
\end{table}
\end{theorem}

\begin{proof}
The matrix $B_{n,p}$ has  two eigenvalues
$\mu_1=\theta_{n-1,p}^2$ and $\mu_2=p^{n-2}$, which are  of multiplicity $d_1=1$ and
$d_2=\theta_{n,p}-1$, respectively.  Lemma
\ref{lem:eigenvalue_recursive_formula} shows that
\begin{itemize}
\item[--] $\lambda_1=p^{(2n-4)(e-1)} \mu_1$
is an eigenvalue of $B_{n,p^e}$ of multiplicity at least $d_1$;
\item[--]   $\lambda_2=p^{(2n-4)(e-1)} \mu_2$ is an eigenvalue of $B_{n,p^e}$ of multiplicity at least $d_2$.
\end{itemize}
 Lemma  \ref{lem:small_eig} shows that  $\mu_s=p^{(s-1)(n-2)}$ is an eigenvalue of
$B_{n,p^{s-1}}$ of multiplicity
at least
$d_s=(p^{n-1}-1)\theta_{n,p^{s-2}}$ for every $s\in \{3,\ldots, e+1\}$.
Due to Lemma \ref{lem:eigenvalue_recursive_formula},
\begin{itemize}
\item[--]
$\lambda_s=p^{(2n-4)(e-s+1)}\mu_s$  is an eigenvalue of
$B_{n,p^e}$
of  multiplicity at least $d_s$ for $s\in \{3,\ldots, e+1\}.$
\end{itemize}
The sum of the multiplicities of  $\lambda_1,\ldots, \lambda_{e+1}$ is at least
$$\sum_{s=1}^{e+1} d_s=1+\theta_{n,p}-1+\sum_{s=3}^{e+1}\left(p^{n-1}-1\right)\theta_{n,p^{s-2}}=
\theta_{n,p^e}$$
 Hence,  the multiplicity  of $\lambda_s$ must be  $d_s=(p^{n-1}-1)\theta_{n,p^{s-2}}$ for
every $s\in[e+1]$.
\end{proof}

\section{The General Case}

In this section, we determine the eigenvalues of $B_{n,m}$ for a general integer $m$. Firstly, we show a tensor
lemma on the matrix $B_{n,m}$.

\begin{lemma}\label{pro:tensorp}
If $m=m_1m_2$ for two coprime integers $m_1$ and $m_2$, then we have that
\begin{equation}
B_{n,m}\sim B_{n,m_1} \otimes
B_{n,m_2}.
\end{equation}
\end{lemma}

\begin{proof}
Let $\pi: \mathbb{P}_{n,m_1}\times \mathbb{P}_{n,m_2}
\rightarrow \mathbb{P}_{n,m}$
be the mapping defined by $\pi(u,v)=w$,
where
\begin{equation}
w_i\equiv u_i\bmod m_1 {\rm~and~} w_i\equiv v_i\bmod m_2
 \end{equation}
 for every $i\in[n]$. The $\pi$ is well-defined.
In fact, let $w=\pi(u,v)$ and $
w^\prime=\pi(u^\prime, v^\prime)$, where
$u, u^\prime\in \mathbb{P}_{n,m_1}$ and $v, v^\prime\in \mathbb{P}_{n,m_2}$.
If $u\sim u^\prime$ and $v\sim v^\prime$, then there are integers $\lambda\in \mathbb{Z}_{m_1}^*$ and $\mu\in \mathbb{Z}_{m_2}^*$ such that
 \begin{equation}
u^\prime_i\equiv \lambda u_i\bmod m_1 {\rm ~and~} v^\prime_i\equiv \mu v_i\bmod m_2.
  \end{equation}
for every $i\in [n]$.
Let $\delta$ be an integer such that
\begin{equation}
\delta\equiv \lambda \bmod m_1{\rm~and~}\delta \equiv \mu \bmod m_2
\end{equation}
Due to (30), (31), (32), we have that $w^\prime _i\equiv \delta  w_i \bmod m$
for every $i\in [n]$. Hence, $w^\prime\sim w$.

It is easy to see that the mapping $\pi$ is bijective and $\theta_{n,m}=\theta_{n,m_1}\theta_{n,m_2}$.
Hence, $\pi$ is bijective. Let
$w,w^\prime$ be defined as above. Clearly,
$\langle w,w^\prime\rangle\equiv 0 \bmod m$ if and only if
$\langle u,u^\prime\rangle\equiv 0 \bmod m_1$ and
$\langle v,v^\prime\rangle\equiv 0 \bmod m_2$.
Hence, the $(w,w^\prime)$ entry of $A_{n,m}$ is equal to 1 if and only if the
$(u,u^\prime)$ entry of
$A_{n,m_1}$ and the $(v,v^\prime)$ entry of $A_{n,m_2}$ are both equal to 1. Hence,
$A_{n,m}\sim A_{n,m_1}\otimes A_{n,m_2}$. Thus

\begin{equation*}
\begin{split}
B_{n,m}
&=A_{n,m}A_{n,m}^t\\
&\sim (A_{n,m_1}\otimes A_{n,m_2}) (A_{n,m_1}\otimes A_{n,m_2})^t\\
&=  (A_{n,m_1}\otimes A_{n,m_2}) (A_{n,m_1}^t\otimes A_{n,m_2}^t)\\
&=  (A_{n,m_1}A_{n,m_1}^t)\otimes (A_{n,m_2} A_{n,m_2}^t)\\
&= B_{n,m_1}B_{n,m_2}
\end{split}
\end{equation*}
which is the expected result.
\end{proof}

As an immediate corollary of Lemma 4.1, we have

\begin{corollary}\label{thm:B_m_n_decomposition}
Let $m=m_1\cdots m_r=p_1^{e_1}\cdots p_r^{e_r}$ for distinct primes $p_1,\ldots, p_r$ and
positive integers
$e_1,\ldots, e_r$,
where $m_s=p_s^{e_s}$ for every $s\in [r]$.
Then
$B_{n,m}\sim B_{n,m_1} \otimes \cdots \otimes B_{n,m_r}$.
\end{corollary}

\begin{lemma}\label{lem:eigenvector_form}
Let $\lambda$ be an eigenvalue of $B_{n,p^e}$
of  multiplicity $d$. Then there are $d$ eigenvectors (column vectors)
${\bf u}_1,\ldots, {\bf u}_d$ of
$B_{n,p^e}$  with eigenvalue $\lambda$ such that
\begin{equation}
({\bf u}_1,\ldots, {\bf u}_d)\sim
\begin{pmatrix}
I_d\\ *
\end{pmatrix}.
\end{equation}
\end{lemma}

\begin{proof}
If $e = 1$, then the single column of $J_{l\times 1}$ is an eigenvector of $B_{n,p}$
with eigenvalue $\theta_{n-1,p}^2$, where
$l=\theta_{n,p}$. Furthermore, there are $d=l-1$ eigenvectors ${\bf u}_1,\ldots, {\bf u}_d$
with eigenvalue $p^{n-1}$ such that
$$
\left({\bf u}_{1},\ldots, {\bf u}_{d}\right)=
\begin{pmatrix}
I_d \\
-J_{1\times d}
\end{pmatrix}.
$$

We give the proof for $e > 1$ by induction. Lemma 3.1 shows that $B_{n,p^e}$ has $e + 1$
 eigenvalues. We
prove for each eigenvalue. Firstly, due to the proof of Lemma 3.12, we have that
\begin{itemize}
\item [--] the eigenvectors of $B_{n,p^e}$ with eigenvalue $\lambda_{e+1}$ have the form (33).
\end{itemize}

Secondly, for every $s\in [e]$, let ${\bf u}_1,\ldots, {\bf u}_{d_s}$ be the eigenvectors of
$B_{n,p^{e-1}}$ with eigenvalue $\lambda_s\cdot p^{-2(n-2)}$.
Due to the induction hypothesis, we have that
$$
\left({\bf u}_{1},\ldots, {\bf u}_{d_s}\right)=
\begin{pmatrix}
I_{d_s} \\
*
\end{pmatrix}.
$$
Let ${\bf v}_i=J_{p^{n-1}\times 1}\otimes {\bf u}_i$ for every $i\in [d_s]$. Lemma 3.10 shows that
${\bf v}_1, \ldots, {\bf v}_{d_s}$ are eigenvectors of $B_{n,p^e}$
with eigenvalue $\lambda_s$. Thus, it is easy to see that
\begin{itemize}
\item[--] the eigenvectors of $B_{n,p^e}$ with eigenvalue $\lambda_s$ have the form (33) for
every $s\in [e]$.
\end{itemize}

Finally, by the induction, we have that (33) holds for $B_{n,p^e}$.
\end{proof}

 \begin{theorem}\label{pro:eig_B_m_n}
{\em (General Case)}
Let $m=m_1\cdots m_r=p_1^{e_1}\cdots p_r^{e_r}$ for distinct primes $p_1,\ldots, p_r$ and
positive integers
$e_1,\ldots, e_r$,
where $m_s=p_s^{e_s}$ for every $s\in [r]$.
Let $\lambda_{s}$ be an eigenvalue of $B_{n,m_s}$ of
multiplicity $d_{s}$ for every
$s\in[r]$. Then $\lambda_1\cdots\lambda_r$
is an  eigenvalue of  $B_{n,m}$ of multiplicity
$d_{1}\cdots d_{r}$.
\end{theorem}

\begin{proof}
Lemma 4.2 shows there are eigenvectors ${\bf u}_{s1},\ldots, {\bf u}_{sd_s}$
of $B_{n,m_s}$ with eigenvalue $\lambda_s$ such that
\begin{equation}
\left({\bf u}_{s1},\ldots, {\bf u}_{sd_s}\right)=
\begin{pmatrix}
I_{d_s} \\
*
\end{pmatrix}.
\end{equation}
Due to Lemma 2.1, we have that
\begin{equation*}
\begin{split}
B_{n,m} \left({\bf u}_{1f_1}\otimes \cdots\otimes
{\bf u}_{rf_r}\right)&=
(B_{n,m_1}\cdot {\bf u}_{1f_1}) \otimes
\cdots
(B_{n,m_r}\cdot {\bf u}_{rf_r})=
(\lambda_1\cdots \lambda_r) \cdot
({\bf u}_{1f_1}\otimes \cdots\otimes
{\bf u}_{rf_r}),
\end{split}
\end{equation*}
where $1\leq f_s\leq d_s$ for every $s\in [r]$. Hence,
$({\bf u}_{1f_1}\otimes \cdots\otimes
{\bf u}_{rf_r})$ is an eigenvector of
$B_{n,m}$  with eigenvalue $\lambda_1\cdots \lambda_r$. Due to (34),
it is not hard to see that the following eigenvectors
$$\{\left({\bf u}_{1f_1}\otimes \cdots\otimes
{\bf u}_{rf_r}\right): 1\leq f_s\leq d_s{\rm~for~}s\in [r]\}$$
are linearly independent. Hence,
$\lambda_1\cdots\lambda_r$ is an eigenvalue of $B_{n,m}$ of
of multiplicity at least $d_1\cdots d_r$.

Let $\lambda_{sj}$ be the eigenvalue of $B_{n,m_s}$
 of multiplicity $d_{sj}$ for every
$s\in[r]$ and $j\in[e_s+1]$.
Then  $\lambda_{1j_1}\cdots\lambda_{rj_r}$
is an  eigenvalue of $B_{n,m}$ of  multiplicity
at least
$d_{1j_1}\cdots d_{rj_r}$.
Theorem \ref{pro:eig_B_p_e_n} shows that
 $$\sum_{j=1}^{e_s+1}d_{sj}
=\theta_{n,m_s}$$ for every $s\in[r]$.
It follows that
\begin{equation*}
\sum_{j_1\in [e_1+1],\ldots, j_r\in [e_r+1]}d_{1j_1}\cdots d_{rj_r}
=\sum_{j_1=1}^{e_1+1}d_{1j_1}\cdot
\sum_{j_2=1}^{e_2+1}d_{2j_2} \cdots  \sum_{j_r=1}^{e_r+1}d_{rj_r}
=
\prod_{s=1}^r\theta_{n,m_s}=\theta_{n,m}.
\end{equation*}
Hence, the multiplicity of $\lambda_1\cdots \lambda_r$
cannot be greater than $d_1\cdots d_r$. In other words, the multiplicity of
$\lambda_1\cdots \lambda_r$ must be exactly $d_1\cdots d_r$.
\end{proof}

\end{document}